\documentclass{article}
\usepackage[margin=1in]{geometry}
\usepackage{graphicx}
\usepackage{amsmath}
\usepackage{amsthm}
\usepackage{amssymb}
\usepackage{physics}
\usepackage{xcolor}
\usepackage{diagbox}
\usepackage{xspace}
\usepackage{multirow}
\usepackage{hyperref}
\usepackage{subcaption}

\definecolor{darkred}{RGB}{139, 0, 0}

\newcommand{\queryamount}[0]{\ensuremath{Q}\xspace}
\newcommand{\abilityamount}[0]{\ensuremath{P}\xspace} 

\newcommand{\nonadaptivelabel}[0]{\ensuremath{nq}\xspace}
\newcommand{\pdqporacle}[0]{\ensuremath{\mathcal{Q}_P}\xspace}
\newcommand{\quantumstate}[1]{\ensuremath{\ket{\phi_{#1}}}\xspace}
\newcommand{\unitarygate}[2][]{\ensuremath{U^{#1}_{#2}}\xspace}
\newcommand{\measurementgate}[2][]{\ensuremath{M^{#1}_{#2}}\xspace}
\newcommand{\parallelregisters}[0]{\ensuremath{r_i}\xspace}
\newcommand{\kinputs}[1]{\ensuremath{#1}-inputs}

\newcommand{\compclass}[2][]{\ensuremath{\mathsf{#2}^{#1}}\xspace}
\newcommand{\pdqpname}[0]{\compclass{PDQP}}
\newcommand{\pdqpnaqname}[0]{\compclass[\nonadaptivelabel]{PDQP}}
\newcommand{\bqpname}[0]{\compclass{BQP}}
\newcommand{\postbqpname}[0]{\compclass{PostBQP}}
\newcommand{\dqpname}[0]{\compclass{DQP}}
\newcommand{\copybqpname}[0]{\compclass{CBQP}}

\newtheorem{theorem}{Theorem}[section]
\newtheorem{corollary}[theorem]{Corollary}
\newtheorem{lemma}[theorem]{Lemma}
\newtheorem{definition}[theorem]{Definition}

\newtheorem{proposition}[theorem]{Proposition}
\newtheorem{result}[theorem]{Result}

\usepackage{enumitem}
\usepackage{tikz}
\usetikzlibrary{quantikz2}
\title{New Lower-bounds for Quantum Computation with Non-Collapsing Measurements\footnotetext{This work is supported by US
Department of Energy (grant no DE-SC0023179) and partially supported by US National Science Foundation
(award no 1954311).}\footnotetext{We want to thank Kunal Marwaha for helpful discussions. We thank the anonymous reviewers for their valuable comments and suggestions.}}

\author{David Miloschewsky\\\small{Department of Computer Science, }\\
\small{Stony Brook University}\\
\small{dmiloschewsk@cs.stonybrook.edu}\and Supartha Podder
\\
\small{Department of Computer Science, }\\
\small{Stony Brook University}\\\small{supartha@cs.stonybrook.edu}}

\date{}

\begin{document}

\maketitle

\begin{abstract}
Aaronson, Bouland, Fitzsimons and Lee~\cite{space_above_bqp} introduced the complexity class \pdqpname (which was original labeled \compclass{naCQP}), an alteration of \bqpname enhanced with the ability to obtain non-collapsing measurements, samples of quantum states without collapsing them. Although $\compclass{SZK}\subseteq \pdqpname$, it still requires $\Omega(N^{1/4})$ queries to solve unstructured search.

We formulate an alternative equivalent definition of \pdqpname, which we use to prove the positive weighted adversary lower-bounding method, establishing multiple tighter bounds and a trade-off between queries and non-collapsing measurements. We utilize the technique in order to analyze the query complexity of the well-studied majority and element distinctness problems. Additionally, we prove a tight $\Tilde{\Theta}(N^{1/3})$ bound on search. Furthermore, we use the lower-bound to explore \pdqpname under query restrictions, finding that when combined with non-adaptive queries, we limit the speed-up in several cases.
\end{abstract}

\section{Introduction}

\subsection{Background and motivation}

The study of the quantum query complexity model has lead to a rewarding avenue of research in understanding quantum speed-ups~\cite{shor_algorithm, bbbv_97, polynomial_method, grover_algorithm}. In this model, one is given access to an input $x$ through a black-box and the goal is to evaluate $f(x)$ for some function $f$ while minimizing the number of quantum queries to $x$. A variety of tools have been developed to lower-bound quantum query complexity, notably the polynomial method~\cite{polynomial_method}, where one uses the fact that a polynomial may describe the amplitudes of a quantum state, the hybrid method~\cite{bbbv_97} and its generalization the adversary method~\cite{adversary_original, adversary_methods_equivalent}, which quantifies the difficulty of distinguishing between inputs with different images. The study and application of these techniques has been instrumental in understanding the complexity class characterizing quantum computation, \bqpname  (Bounded-error Quantum Polynomial-time). Specifically, as one may be given access to $x$ through an oracle $O$, quantum query complexity studies \compclass[O]{BQP} and its relations to other complexity classes with access to $O$.

One of the goals of studying \bqpname is to characterize the classes of functions which may be efficiently computable once fault-tolerant quantum computers have been realized. The difficulty of this realization and the need to mimic our current hardware has lead us to place various restrictions on \bqpname . One prominent issue, the presence of noise, motivated the study of noise in the quantum query complexity setting, i.e., to introduce noise to the oracle calls. This was explored in~\cite{ addendum_noisy_oracle_search}, where the authors showed that Grover's algorithm loses its speed-up even if only a single qubit is randomly affected with depolarizing noise. Furthermore,~\cite{nisq_computation} defined the class \compclass{NISQ} (Noisy Intermediate-Scale Quantum), the class of problems solved by a \compclass{BPP} algorithm with access to a noisy quantum device, emulating some aspects of today's state of quantum computation. They showed an oracle separation between both \compclass{BPP} and \bqpname, and showed that \compclass{NISQ} algorithms cannot solve search as fast as \bqpname. 
These results point to the conclusion that the presence of noise considerably weakens the computational power.

Another intriguing limited setting of the query model is the non-adaptive query model~\cite{nonadaptive_weighted_adversary}. Here we restrict the model to a single round a parallel queries and analyze the necessary width of the round, imposing that the queries are independent of each other. It was shown in~\cite{nonadaptive_weighted_adversary} that this restriction thwarts the quantum speed-up in solving the search and element distinctness problems, implying that adaptivity is necessary in order to obtain polynomial speed-ups for some problems.

 On the other hand, augmenting \bqpname with a metaphysical ability have also been a fruitful avenue in understanding the power of computation. 
 One of the first considered variations of quantum mechanics was adding post-selection to quantum computation, \postbqpname (Post-selected \bqpname). Post-selection is the ability that, given a quantum state, we may \emph{select} a partial measurement outcome, regardless of its probability. It was shown that $\postbqpname=\compclass{PP}$~\cite{pp_postbqp}, establishing an equivalency between a quantum and classical computational class. However, \postbqpname is well above \bqpname as even \compclass{QMA}, a quantum version of \compclass{NP}, is contained in \compclass{PP}.

One of the theories explaining quantum mechanics is the hidden variable theory, an interpretation where the system is described by a state and a hidden variable which stochastically determines the measurement results. \dqpname (Dynamical Quantum Polynomial-time) is defined as the set of languages solved by a classical deterministic polynomial-time Turing machine which can process the history of the hidden variables of a quantum algorithm~\cite{dqp}. At the moment, we only know that $\bqpname\subseteq\dqpname\subseteq\compclass{EXP}$ and the best known search algorithm uses $O(N^{1/3})$ queries. However,~\cite{space_above_bqp} showed that adjusting \dqpname to \compclass{CDQP}(Circuit-indifferent \dqpname), adding a restriction to the hidden variable theory, gives us a $\Omega(N^{1/4})$ lower-bound on search.

A well-studied offspring of \dqpname is \pdqpname (Product Dynamical Quantum Polynomial-time), quantum computation powered with non-collapsing measurements. By non-collapsing measurements, we mean that given an arbitrary state $\ket{\psi}$, we sample a measurement result without collapsing the state.
We may view this as a simplification of \dqpname as we are given samples from the history of computation. Defined in~\cite{space_above_bqp}, the authors proved a $\Omega(N^{1/4})$ lower-bound on search and $\compclass{SZK}\subseteq\pdqpname\subseteq\compclass[\mathsf{PP}]{BPP}$.

A class related to \pdqpname is \copybqpname (Copy \bqpname), \bqpname with the ability to create an unentangled copy of any state. Note that \copybqpname can simulate non-collapsing measurement by creating a copy and measuring it. However, the ability to entangle copies together adds considerable power. It was shown that \copybqpname can solve any \compclass{NP}-hard problem in polynomial time using a single query and $O(\log (N))$ copies~\cite{grover_search_no_signaling_principle}. The class is positioned between  \compclass[\mathsf{PP}]{BQP} and \compclass{PSPACE}~\cite{rewindable_bqp}.

In our work, we revisit \pdqpname as, except for \compclass{CDQP}, it is known to be only slightly stronger than \bqpname as it may at most attain a polynomial speed-up over Grover's algorithm. To understand the lower-bound, we must first describe a caveat. In the standard quantum query model, the goal is to minimize the amount of queries \queryamount. However, in models with an additional ability, such as creating a copy or performing a non-collapsing measurement, we must provide bounds on the uses of the ability, which we denote by \abilityamount. In the case of non-collapsing measurements, without minimizing both \queryamount and \abilityamount, every problem of size $N$ could be trivially solved using a single query and followed by exponentially-many non-collapsing measurements which are repeated until we learn enough information about the input. Furthermore, unrestricted non-collapsing measurements lead to quantum cloning, faster than light communication and constant communication complexity~\cite{space_above_bqp}. Therefore the measure of complexity is $\queryamount + \abilityamount$, the sum of the number of queries and number of non-collapsing measurements.

The search lower-bound may lead some to the conclusion that a \pdqpname algorithm with $\abilityamount$ non-collapsing measurements may be simulated by \bqpname with a squared overhead by running $\abilityamount$ separate computations in parallel. However, this intuition fails as the \emph{deferred measurement principle} does not hold in \pdqpname. To explain why, consider the collision problem. Here, assuming $N\in \mathbb{N}$, one is given black-box access to a function $f:[N] \rightarrow [N]$ with the promise that for any $x\in [N]$, $k = |\{y: f(y)=f(x)\}|$ is either 1 or 2. The algorithm is as follows. First, we query all the indices into a superposition and call the oracle, obtaining,
\begin{align*}
    \ket{\psi^\prime}=\sum_{x\in [N]} \ket{x}\ket{f(x)}
\end{align*}
Next, we measure the query output register, collapsing $\ket{\psi^\prime}$ to $\ket{\psi}$ and obtaining $f(x_1)$ as our measurement result. We find ourselves with the following state,
\begin{align*}
    \ket{\psi} = \begin{cases}
        \ket{x_1}\ket{f(x_1)} &\text{if }k=1\\
        \frac{\ket{x_1} + \ket{x_2}}{\sqrt{2}}\ket{f(x_1)} &\text{if }k=2
    \end{cases}
\end{align*}

By repeating $O(1)$ non-collapsing measurements of $\ket{\psi}$, we find $k$ with high probability, solving the problem with both constant queries and non-collapsing measurements. On the other hand, without the partial measurement, the algorithm is no longer efficient. In contrast, $\Tilde{\Theta}(N^{1/3})$ queries are optimal in \bqpname~\cite{collision_problem_algorithm, collision_element_distinctness_lower_bound}. Note that the proof of $\compclass{SZK}\subseteq \pdqpname$ uses the same partial measurement trick to solve the Statistical Difference problem, which is $\compclass{SZK}\text{-hard}$~\cite{szk_hard_problem}.

When it was introduced, \pdqpname was referred to as \compclass{naCQP} (non-adaptive Collapse-free Quantum Polynomial-time). The class \compclass{naCQP} is a restriction of \compclass{CQP}, the class with the ability to adapt the quantum computation based on the non-collapsing measurement results. The power of this class is unknown as there are no known algorithms which utilize the adaptive feature nor do we have any lower-bounds. We note that that the phrase \emph{non-adaptive} in \compclass{naCQP} is not related to non-adaptive queries discussed above, \compclass{naCQP} has the ability to perform a query at any step of the algorithm.

Recently, non-collapsing measurements gained renewed interest. It was shown that the complexity class \compclass{PDQMA} (Product Dynamical Quantum Merlin-Arthur), the quantum counter-part of \compclass{NP} with the power of non-collapsing measurements, equals \compclass{NEXP}~\cite{pdqma_equals_nexp, pdqma_single}. Inspired by their work and several open directions in the original \pdqpname literature, we revisit the class. 
There are several settings to study \pdqpname in. One may focus on adapting unitaries based on non-collapsing measurement results, restricting the query adaptivity or limiting partial measurements. We explore the power of non-collapsing measurements with non-adaptive queries. To avoid any confusion, for the remainder of the paper, any discussion of non-adaptability will be concerned with queries.

It is important to stress that the point of these adjustments is not to argue for alterations of the postulates of quantum mechanics. Instead, we explore quantum computation and complexity from an alternative point. Therefore, the models we define to encompass the adjustments of \bqpname may use ideas which are not inherently present in quantum computation. This is unavoidable as if they could be described solely using the postulates of quantum mechanics, their complexity class would be equal to \bqpname.

\subsection{Results}\label{subsection_results}

\subsubsection{Lower-bounding methods}

We obtain a trade-off between queries and non-collapsing measurements by developing the adversary method lower-bound in our setting. For simplicity, our results are presented using the basic adversary method from~\cite{adversary_original}, but we prove its generalization, the positive weighted adversary method of~\cite{weighted_adversary} (see Subsection~\ref{subsection_proof_application}). When discussing query complexity, we will use the name of the complexity class as the identifier of the model. We note that in all cases, one recovers the original adversary method by setting the number of non-collapsing measurements to 1, essentially only measuring the end state. The result below exhibits that non-collapsing measurements provide us with a limited speed-up over \bqpname in problems which may be lower-bounded by the positive weighted adversary method.

\begin{result}[\textbf{Adversary method on \bqpname with non-collapsing measurements}, directly from Lemma~\ref{lemma_weighted_adversary_partial_measurement}]\label{result_adversary}
    Let $f:\{0,1\}^{|\Sigma|}\rightarrow \{0,1\}$ be a boolean function and $x$ an input.
    Define $X\subseteq f^{-1}(0)$, $Y\subseteq f^{-1}(1), R\subseteq X\times Y$. Based on $X,Y, R$, let $m, m^\prime$ be the minimum number of elements $(x,y)$ for each $x\in X$ and $y\in Y$ respectively and $l,l^{\prime}$ are the most elements $(x,y)$ where $x(i)\neq y(i)$ for each $i\in\Sigma$ and for each $x\in X$ and $y\in Y$ respectively.

    Any algorithm solving $f(x)$ in a \pdqpname setting with \queryamount queries and \abilityamount non-collapsing measurements satisfies,
    \begin{align*}
        \queryamount \abilityamount = \Omega \left( \sqrt{\frac{mm^\prime}{ll^\prime}} \right)
    \end{align*}
\end{result}

We develop a similar trade-off between \queryamount and \abilityamount in the non-adaptive query setting. Defined in Subsection~\ref{subsection_nonadaptive_model}, we denote it with the superscript \nonadaptivelabel.

\begin{result}[\textbf{Adversary method on non-adaptive \bqpname with non-collapsing measurements}, directly from Lemma~\ref{lemma_nonadaptive_final}]
    Define $f, x, m, m^\prime, l, l^\prime$ as in Result~\ref{result_adversary}. Any algorithm solving $f(x)$ in a \pdqpnaqname setting with \queryamount queries and \abilityamount non-collapsing measurements satisfies,
    \begin{align*}
        \queryamount \abilityamount = \Omega \left(\max \left\{ \frac{m}{l}, \frac{m^\prime}{l^\prime} \right\} \right)
    \end{align*}
\end{result}

Additionally, we obtained a variation of Result~\ref{result_adversary} for \copybqpname by applying a similar proof technique.

\begin{result}[\textbf{Adversary method on \bqpname with the ability to copy states}, directly from Lemma~\ref{lemma_copy_adversary}]\label{result_copy}
    Define $f, x, m, m^\prime, l, l^\prime$ as in Result~\ref{result_adversary}. Any algorithm solving $f(x)$ in a \copybqpname setting with \queryamount queries and \abilityamount copies satisfies,
    \begin{align*}
        Q \cdot 2^P = \Omega \left( \sqrt{\frac{mm^\prime}{ll^\prime}} \right)
    \end{align*}
\end{result}

\subsubsection{Applications of the lower-bounding methods}

\begin{table}[t!]
    \centering
    \renewcommand{\arraystretch}{1.3}

\begin{tabular}{|c|c|c|c|}
    \hline
    \diagbox{Problem}{Setting}    &  \bqpname & \pdqpname & \pdqpnaqname  \\
        \hline
    \rule[0pt]{0pt}{\heightof{A}+1.5ex} \textbf{Unstructured search}      & $\Tilde{\Theta}(\sqrt{N})$ & $\color{darkred}\Tilde{\Theta}(N^{1/3})$ & $\color{darkred}\Tilde{\Theta}(\sqrt{N})$  \\ \hline
    \textbf{Majority/Parity}    & $\Theta(N)$ & $\color{darkred}\Tilde{\Theta}(\sqrt{N})$ & $\color{darkred}\Tilde{\Theta}(\sqrt{N})$   \\\hline
     \rule[0pt]{0pt}{\heightof{A}+1.5ex}\textbf{Collision}   & $\Tilde{\Theta}(N^{1/3})$ & $\Theta(1)$ & $\Theta(1)$ \\ \hline
    \multirow{2}{*}{\textbf{Element distinctness}} & \multirow{2}{*}{$\Tilde{\Theta}(N^{2/3})$} & $\color{darkred}O(\sqrt{N})$ &  \multirow{2}{*}{$\color{darkred}\Tilde{\Theta}(\sqrt{N})$}\\
     & & $\color{darkred}\Omega(N^{1/4})$  & \\ \hline
\end{tabular}
    \caption{Summary of bounds on all the computational models explored. The non-trivial bounds we establish or improve in this paper are colored \textcolor{darkred}{red}. The $\mathsf{BQP}$ bounds are added for reference.}
    \label{table_bounds}
\end{table}

We apply the results above to specific problems in order to draw inference about \pdqpname. Specifically, we consider the bounds on search, majority and element distinctness in the settings discussed above. A summary of the bounds is in Table~\ref{table_bounds}. Proof of the applications of the bounds above may be found in Subsection~\ref{subsection_proof_application}.

In the search problem, we are given access to a function $f:[N]\rightarrow \{0,1\}$ and the goal is to decide whether there exists an $x\in [N]$ such that $f(x)=1$. The lower-bound on search in \pdqpname was one of the main results in~\cite{space_above_bqp}, but it was not tight. As the method in Result~\ref{result_adversary} provides us with the same lower-bound, we prove a tighter bound in Lemma~\ref{lemma_lowerbound_search}.

\begin{result}[\textbf{Bounds on search}]
    We obtain tight bounds for unstructured search. In \pdqpname, the bound is $\Tilde{\Theta}(N^{1/3})$, while in \pdqpnaqname it is $\Tilde{\Theta}(\sqrt{N})$.
\end{result}

Next, we consider the majority problem. In this problem, assuming $n\in \mathbb{N}$ and $N=2^n$, we are given $f:[N]\rightarrow \{0,1\}$ and must decide whether $s=|\{x: f(x)=1\}| \geq 2^{n-1}$. As majority may be solved efficiently in \compclass{PP}, we obtain an oracle separation between \compclass{PP} and \pdqpname.\footnote{The lower-bound on search implies a separation between \compclass{NP} and \pdqpname, therefore implying this separation. However, we don't expect majority to separate \pdqpname and \compclass{NP} nor \compclass{QMA}.} Furthermore, the parity problem, where one must decide $\oplus_i x_i$, has the same exact query complexity. This implies the existence of an oracle $O$ such that $\compclass[O]{\oplus P}\not\subseteq \compclass[O]{PDQP}$.

\begin{result}[\textbf{Bounds on majority}]
    We obtain tight bounds for the majority and parity problems. In both \pdqpname and \pdqpnaqname, they are $\Tilde{\Theta}(\sqrt{N})$.
\end{result}

While the collision problem has $O(1)$ bounds, it is interesting that the same is not the case for the element distinctness problem. The problem is to decide, given $f:[N]\rightarrow [N]$, whether there exist $x,y$ such that $x\neq y$ and $f(x)=f(y)$. In \bqpname, we have a clear connection between the two classes, as their bounds may be derived from each other~\cite{childs_lecture_notes}. However, this is not the case in any \pdqpname setting. It would be very interesting to characterize what causes this difference. 

\begin{result}[\textbf{Bounds on element distinctness}]
    We obtain bounds for the element distinctness problem. In \pdqpname, the bound is between $O(\sqrt{N})$ and $\Omega(N^{1/4})$, while in \pdqpnaqname it is $\Tilde{\Theta}(\sqrt{N})$.
\end{result}

Finally, let us consider the \copybqpname setting. ~\cite{grover_search_no_signaling_principle} showed an algorithm which uses a single query and $O(\log(N))$ copies to solve the search problem. Applying Result~\ref{result_copy} shows that this algorithm is optimal.

\subsection{Proof techniques}

We describe the ideas behind our proofs. The key part of this paper is Definition~\ref{definition_pdqp_our}, an alternative definition of the \pdqpname model. It allowed us to prove Lemma~\ref{lemma_fidelity_partial_measurement}, enabling us to ignore partial measurement in the proofs of our lower-bounds.

\begin{figure}
    \centering
    \input{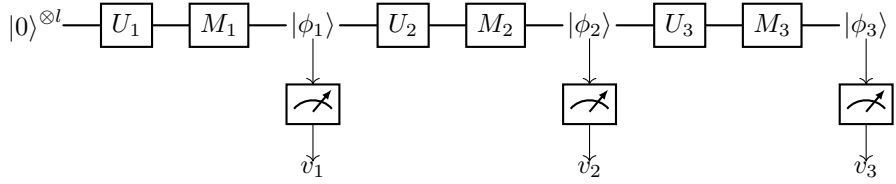}
    \caption{Example of the process of the oracle $\mathcal{Q}_P$ with $P=3$. Given a circuit $C$, the oracle $\mathcal{Q}_P$ returns $\{v_i\}_i^3$ by measuring the states $\{\phi_i\}_i^3$ at each step of $C$.}\label{figure_pdqp_original_definition}
\end{figure}

\subsubsection{Alternative definition of computation with non-collapsing measurements}

Before explaining our alternative definition, we discuss the original definition of the model with non-collapsing measurements from~\cite{space_above_bqp}. \pdqpname is defined as the set of languages solvable by a deterministic Turing machine with one query to an oracle \pdqporacle, which runs a quantum circuit $C$ with \abilityamount steps. At step $i$, when the computation is in the state \quantumstate{i}, the oracle obtains a measurement sample $v_i$ of \quantumstate{i} without collapsing it. At the end, the machine receives $\{v_i\}_i$ and processes it. An example of how the oracle runs $C$ is in Figure~\ref{figure_pdqp_original_definition}. We define an alternative oracle $\pdqporacle^*$ which inputs circuit $C$ and outputs $\{v_i\}_i$ as well. However, we explicitly describe the quantum circuit $C^*$ which $\pdqporacle^*$ runs. The goal is to describe $C^*$ such that it outputs the tensor product of each state before the non-collapsing measurement. Therefore, we may obtain the results by measuring the end state.

Assume the input circuit $C$ is an alternating combination of unitaries \unitarygate{i} and (potentially empty) partial measurement gates \measurementgate{i}. We create the variants \unitarygate[*]{i} and \measurementgate[*]{i} of these gates in $C^*$. The goal is that at every step $i$, we have $r_i = \abilityamount - i + 1$ copies of the state \quantumstate{i} and our gates affect $r_i$ states, each of whose goal is to end in \quantumstate{j} where $j\geq i$. An example of the circuit $C^*$ is in Figure~\ref{figure_pdqp_definition_example}. The unitary \unitarygate[*]{i} is a parallel application of \unitarygate{i} over each of the $r_i$ registers separately. The measurement \measurementgate[*]{i} is more complex as it must ensure that the partial measurement results are the same for all $r_i$ registers in order for them to remain copies of each other. We define it as a gate which may post-select on measuring the same outcome for each register with an adjusted probability distribution.

Let us explain the probability adjustment. Suppose we are measuring a single qubit in each register and want to ensure that we measure $\ket{0}^{\otimes r_i}$ or $\ket{1}^{\otimes r_i}$. Let the probability of measuring 1 at step $i$ of $C$ be $a$. If we post-selected on this to happen in a \postbqpname manner, the probability of measuring $1^{\otimes r_i}$ is $\frac{a^{r_i}}{a^{r_i}+(1-a)^{r_i}}$, not $a$ as desired. Therefore we must adjust the probabilities of measuring $1^{\otimes r_i}$ by $d = a^{r_i-1}$.\footnote{This is the reason why it is difficult to directly show $\pdqpname \subseteq \postbqpname$.} It is important to mention that \measurementgate[*]{i} is not a valid quantum operator as its projectors don't sum to the identity. The only measurement which is allowed by quantum mechanics is the measurement of the states at the end of $C^*$. This is why \pdqpname is more powerful than \bqpname.

\subsubsection{Obtaining the lower-bounding techniques}

Using our alternative definition, we show in Lemma~\ref{lemma_fidelity_partial_measurement} that given two arbitrary states $(\rho_{1,i},\rho_{2,i})$ created by $C^*$, applying \measurementgate[*]{i} cannot decrease the fidelity between them. Our adversary method uses fidelity as the progress measure, implying that we may ignore the effect of partial measurement during our proof. Finally, we obtain a positive weighted adversary method~\cite{weighted_adversary} for our model by closely following the original proof while accounting for the presence of copies of states, obtaining Lemma~\ref{lemma_weighted_adversary_partial_measurement}.

Our remaining lower-bounds use the lemmata we proved above. First, our result on non-adaptive queries with non-collapsing measurements is described in Lemma~\ref{lemma_nonadaptive}. The proof combines the original adversary method proof for non-adaptive queries~\cite{nonadaptive_weighted_adversary}, replacing the original positive weighted adversary with Lemma~\ref{lemma_weighted_adversary_partial_measurement}. Second, our tight $\Omega(N^{1/3})$ lower-bound on search in Lemma~\ref{lemma_lowerbound_search} is a direct extension of the proof on search in \pdqpname with no partial \emph{collapsing} measurements (Appendix E in~\cite{space_above_bqp}) combined with our ability to ignore the effect of partial measurements which comes directly from Lemma~\ref{lemma_fidelity_partial_measurement}.

Additionally, we apply the same idea of using non-standard gates to describe a computational model in order to obtain a lower-bounding technique for quantum computation with copies, \copybqpname. The techniques remain the same as above, except we allow for the entanglement between copies instead of instantly collapsing them.

\begin{figure}
    \centering
    \input{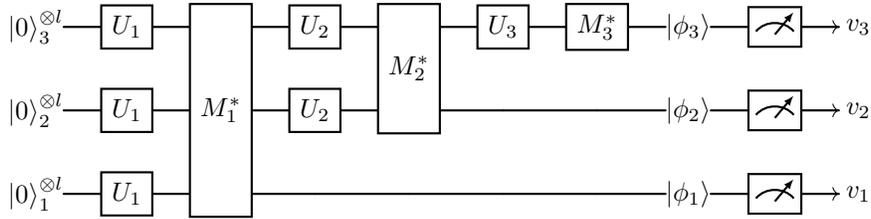}
    \caption{Example of the circuit $C^*$  with $P=3$ which is run by oracle $\mathcal{Q}_P^*$, producing the same results as oracle $\mathcal{Q}_P$.}\label{figure_pdqp_definition_example}
\end{figure}

\subsection{Related work}

Since the introduction of non-collapsing measurements~\cite{space_above_bqp}, there have been several papers analyzing the adjustment in other settings. Specifically, we have $\compclass{PDQP/qpoly}=\compclass{ALL}$~\cite{pdqpqpoly_equals_all} and $\compclass{PDQMA} = \compclass{NEXP}$~\cite{pdqma_equals_nexp, pdqma_single}. Furthermore, it was shown in~\cite{pdqp_oracle_separation} that there exists an oracle $O$ such that $\compclass[O]{P} = \compclass[O]{BQP} = \compclass[O]{SZK} =\compclass[O]{PDQP} \neq (\compclass{UP}\cap \compclass{coUP})^O$ and that relative to a random oracle $O$, $\compclass[O]{NP}\not \subseteq \compclass[O]{PDQP}$. Finally,~\cite{no_collapsing_measurement} studied the effect on quantum theory when we may \emph{only} perform non-collapsing measurements, meaning we may never collapse a quantum state.

Another alteration worth mentioning is adding the power to rewind from a collapsed result back to the state prior to measurement. Denoted \compclass{RwBQP} (Rewindable \bqpname), it was shown to be equal to \copybqpname and \compclass{AdPostBQP} (Adaptive \postbqpname)~\cite{rewindable_bqp}. Adaptivity in \compclass{AdPostBQP} means that based on partial measurement outcomes, we adapt the post-selection. Equivalently, \compclass{RwBQP} may be described as \pdqpname with the additional power to decide whether, once we obtain a non-collapsing measurement, we collapse the state to that sample. This means that adding some form of adaptivity to \postbqpname and \pdqpname makes them equivalent.

Furthermore, \pdqpname without partial measurements may be viewed as \bqpname with parallel queries. Various bounds on this setting have been developed, such as for Grover's algorithm~\cite{grovers_search_optimal} and the general adversary method~\cite{optimal_parallel_query}.

Closely related to parallel queries is adaptivity as non-adaptive algorithms have a single parallel query. Adaptivity of $\mathsf{BQP}$ queries has been a subject of prior research. The positive weighted adversary method was adjusted for non-adaptive queries~\cite{nonadaptive_weighted_adversary} and it was shown that total functions depending on $N$ variables require $\Omega(N)$ queries to the input~\cite{nonadaptive_query_complexity}. Both results showed that non-adaptivity restricts the speed-ups quantum algorithms may achieve. Related to this, the analysis of the number of adaptive queries was studied in~\cite{power_of_adaptivity}, where the authors showed that for any $r\in \mathbb{N}$, the $2r$-fold Forrelation problem can be solved with $r$ adaptive rounds using one query per round, while $r-1$ adaptive round would require an exponential number of queries per round.

\subsection{Open problems}\label{subsection_open_problems}

We list several open questions worth pursuing below.

\begin{itemize}
    \item In \bqpname, the bounds between the collision problem and the element distinctness problem are polynomially related. Our results show that this is not true in \pdqpname. Is it possible to characterize the reason behind this stark difference? We note that it is something which is not recognizable by any \bqpname algorithm.
    \item All efforts to lower-bound the query complexity of \pdqpname used methods derived from \bqpname. Is it possible to develop lower-bounding techniques which are tailored for non-collapsing measurements? Similarly, is it possible to obtain lower-bounds using an altered version of the polynomial method\footnote{We discuss this in Subsection~\ref{subsection_polynomial_method}}?
    \item How does \pdqpname or other metaphysical classes act under a different restrictions. For example, how would it work with noisy computation? Are there restrictions which only affect \bqpname, but not \pdqpname?
    \item What is the power of \pdqpname if it can adapt its unitaries based on the non-collapsing measurement result?
    \item Are there modifications of \bqpname combined with the power of witnesses which are larger than \compclass{NEXP}?
    \item There exist oracle separations between \pdqpname and \compclass{QMA} in both ways\footnote{$\compclass[O]{QMA}\not \subseteq \compclass[O]{PDQP}$ holds by the lower-bound on search, while $\compclass[O]{PDQP}\not \subseteq \compclass[O]{QMA}$ is by an oracle separation between \compclass{SZK} and \compclass{QMA}~\cite{szk_qma_oracle_separation} and $\compclass{SZK}\subseteq \pdqpname$.}, but there are other relationships worth exploring. For example, what is the relationship between $\mathsf{PDQP}$ and advice classes such as $\compclass{BQP/qpoly}$?
\end{itemize}
 
\section{Preliminaries}

We assume familiarity with basic quantum computation concepts; for introductory texts, see~\cite{Nielsen_Chuang, watrous_lecture_notes}.
Let us define the notation used in this paper. Unless mentioned otherwise, $N\in \mathbb{N}$ specifies the size of an arbitrary input. The set $\{1,..,N\}$ is denoted by $[N]$. For a string $s\in \{0,1\}^n$, $s(i)$ denotes the character in position $i$ and the Hamming weight $|s| = |\{s(i): s(i)=1, i\in [n]\}|$. When describing a set $\{a_i\}_i^j$, we implicitly mean $\{a_i\}_{i=1}^j$. On the other hand, $\{a_i\}_i$ indicates a general set enumerated over $i$. The symbols \unitarygate{} and \measurementgate{} denote arbitrary unitary and measurement gates. The Kronecker delta function $\delta_{a,b}$ is a boolean function which is 1 if and only if $a=b$.
A tensor product of $k$ states $\phi$ is denoted by $\phi^{\otimes k}$ in order to avoid confusion with other superscripts.
The symbol $F$ denotes the fidelity measure, while $f$ denotes an arbitrary boolean function. We list several key properties of fidelity.

\begin{itemize}
    \item (\textbf{Uhlmann's Theorem}~\cite{uhlmann_theorem}) Let $\rho_1, \rho_2$ be quantum states and $\ket{\phi_1}$ a fixed purification of $\rho_1$. Then $F(\rho_1, \rho_2) = \max_{\ket{\phi_2}}  \left| \bra{\phi_1}\ket{\phi_2} \right|$ where the maximization is over purifications $\ket{\phi_2}$ of $\rho_2$.
    \item For any states $\rho_1, \rho_2$, $0\leq F(\rho_1, \rho_2)\leq 1$.
    \item (\textbf{Multiplicativity}) For any states $\rho_1, \sigma_1$ and $\rho_2, \sigma_2$, $F(\rho_1\otimes \rho_2, \sigma_1 \otimes \sigma_2 ) = F(\rho_1,\sigma_1 )F(\rho_2,\sigma_2 )$.
    \item (\textbf{Unitary invariance}) For any states $\rho_1, \rho_2$ and an arbitrary unitary $U$, $F(\rho_1, \rho_2) = F(U \rho_1 U^\dagger, U \rho_2 U^\dagger )$.
\end{itemize}

\subsection{Quantum Query Model}

We will be using the quantum query model. In this model, the goal is to compute $f(x)$ for some boolean function and an input $x = (x(0), ...,x(N-1))$. We are given access to $x$ via an oracle $O_x$. We shall assume that we are working with a phase oracle, defined as follows,
\begin{align*}
    O_x: \ket{i,b} \rightarrow (-1)^{x(i)\cdot b}\ket{i,b}
\end{align*}

where $i$ is an query input index of $x$ and $b\in \{0,1\}$. The definition of the quantum query model with non-collapsing measurements is in Section~\ref{section_pdqp_model}. We will use the following lemma.

\begin{lemma}[Hybrid argument of~\cite{bbbv_97}]\label{lemma_hybrid_argument} Let $\ket{\psi_t}$ and $\ket{\psi_t^x}$ be states of an algorithm solving search using $Q$ queries in total after $t$ steps without partial measurement where respectively there is no marked element or $x\in [N]$ is the marked element. Then,
\begin{align*}
    \sum_{x=1}^{N} \norm{\ket{\psi_t} - \ket{\psi_t^x}}_2^2 \leq 4Q^2
\end{align*}
\end{lemma}

The adversary method proof technique is developing an upper-bound on the progress an algorithm makes at one step in order to lower-bound the number of necessary steps. We will use the symbol $\Phi(t)$ as the algorithm progress measure at step $t$. In order to develop our positive weighted adversary lower-bound, we will use the notation from~\cite{weighted_adversary}.

\begin{definition}[Weight scheme of~\cite{weighted_adversary}]\label{definition_weight_scheme}
    Consider an arbitrary boolean function $f$. Let $X\subseteq f^{-1}(0)$, $Y\subseteq f^{-1}(1)$ and $R\subseteq X\times Y$. A weight scheme is defined using $w(x,y)$ and $w^\prime(x,y,i)$ such that for all $(x,y)\in R$ and $i\in [N]$ where $x(i)\neq y(i)$,
    \begin{align*}
        &w(x,y)>0 \\
         &w^\prime(x,y,i)>0 \\
         &w^\prime (y,x,i)>0 \\
         &w^\prime(x,y,i)w^\prime (y,x,i)\geq w^2(x,y)
    \end{align*}
    For a variable $x$, define the weight $wt(x)$ of a variable and the load of a variable $v(x,i)$ as,
    \begin{align*}
        wt(x)&=\sum_{y: (x,y)\in R} w(x,y)\\
        v(x,i) &= \sum_{\substack{y: (x,y)\in R \\ x(i)\neq y(i)}} w^\prime (x,y,i)
    \end{align*}
\end{definition}

\begin{definition}[Weight load of~\cite{weighted_adversary}]\label{definition_weight_load}
    Assume a valid weight scheme as in Definition~\ref{definition_weight_scheme} for an arbitrary $f$. Let $C\in\{X,Y\}$. Then the maximum  $C$-load is defined as,
    \begin{align*}
        v_C = \max_{\substack{x\in C \\ i\in [N]}} \frac{v(x,i)}{wt(x)}
    \end{align*}
    We define the maximum load of $f$ as,
    \begin{align*}
        v_{max} = \sqrt{v_A v_B}
    \end{align*}
\end{definition}

Additionally, we will use the following property from~\cite{weighted_adversary}.
\begin{proposition}[Weight identity, proof of Lemma~4 of~\cite{weighted_adversary}]\label{proposition_weight_identity}
    Assume variables defined as in Definitions~\ref{definition_weight_scheme} and~\ref{definition_weight_load}. Let $\Phi(t) = \sum_{(x,y)\in R} w(x,y) F(\rho_{x,t}, \rho_{y,t})$ and $\alpha_{x,i}$ be the amplitude of $\rho_{x,t}$ associated with the query input index $i$. Then,
    \begin{align*}
        \sum_{\substack{(x,y)\in R \\ i: x(i)\neq y(i)}} 2 w(x,y) \left|\alpha_{x,i} \right| \left|\alpha_{y,i}\right| \leq  v_{max}\Phi(0)
    \end{align*}
\end{proposition}

\section{Non-collapsing measurement model}\label{section_pdqp_model}

We provide a formal definition of the quantum query model with non-collapsing measurements. We begin with the original model from~\cite{space_above_bqp} and follow it with our alteration. At the start, we consider the case without an oracle, which we add at the end.

The original definition from~\cite{space_above_bqp} is as follows. Let \abilityamount be the total number of non-collapsing measurements made and \pdqporacle be a quantum oracle which takes as an input a quantum circuit $C=(U_1, M_1, .., U_{\abilityamount}, M_{\abilityamount})$ where $U_i$ is a unitary operator on $l$ qubits and $M_i$ is a collapsing measurement operator on $s$ qubits such that $0\leq s\leq l$. Let,
\begin{align*}
    \ket{\phi_0} &=\ket{0}^{\otimes l} \\
    \ket{\phi_{i}} &= M_iU_i\ket{\phi_{i-1}}\\
\end{align*}
Define the oracle \pdqporacle as a quantum oracle which inputs $C$ and outputs $\{\ket{\phi_i}\}_i^{\abilityamount}$. If this oracle is called by a machine which cannot hold quantum states, the machine receives the measurement results $\{v_i\}_i^{\abilityamount}$ of $\{\ket{\phi_i}\}_i^{\abilityamount}$ (i.e. upon receiving $\ket{\phi_i}$, the state immediately collapses to $v_i$).

\begin{definition}[Section 2 of~\cite{space_above_bqp}]\label{definition_pdqp_original}
    \pdqpname is the class of languages $L$ which may be recognized by a polynomial-time deterministic Turing machine with one query to \pdqporacle with probability of error at most $1/3$.
\end{definition}

We propose an alternative definition of the computational model by altering the oracle \pdqporacle to an explicit circuit. Let $\parallelregisters=\abilityamount-i+1$. The circuit $C^*$ starts with \abilityamount $l\text{-qubit}$ registers, each initialized at zero. We label each register with index $i$ which denotes which step of $C$, state $\quantumstate{i}$, they compute. At each step of $C^*$, we apply adjusted gates $U^*_i$ and $M^*_i$ over \parallelregisters registers to all registers $j\geq i$, creating $\{\ket{\phi_i}\}_i$ through parallel computation.

Let $a_{i,n}$ be the probability of measuring $n$ from $M_i$ and $d_{i,n} = \frac{1}{a_{i,n}^{\parallelregisters-1}}$. Let,
\begin{align}
    \ket{\psi_0^*} &= \ket{0}^{\otimes l \abilityamount} \\
    U_i^* &= U_i^{\otimes \parallelregisters} \\
    M_i^* &= \sum_n d_{i,n} (\ket{n}\bra{n})^{\otimes \parallelregisters} \label{eq_our_definition_measurement}\\
    \ket{\psi_{i}^*} &= M_i^*U_i^*\ket{\psi_{i-1}^*}
\end{align}

We define $\pdqporacle^*$ as a quantum oracle which inputs $C$, runs $C^*$ and outputs the final state $\ket{\psi_{\abilityamount}^*} = \otimes_i^{\abilityamount} \ket{\psi_i}$. As with \pdqporacle, if the machine calling $\pdqporacle^*$ cannot process quantum states, it receives $\{v_i\}_i^{\abilityamount}$.

\begin{proposition}\label{proposition_equivalent_models}
    \pdqpname is equivalent to the class of languages $L$ which may be recognized by a polynomial-time deterministic Turing machine with one query to $\mathcal{Q}_P^*$ with error probability at most $1/3$.
\end{proposition}

\begin{proof}
    By Definition~\ref{definition_pdqp_original}, $\pdqpname=\compclass[\pdqporacle,1]{BPP}$ and we want to show $\pdqpname=\compclass[\pdqporacle^*,1]{BPP}$. As both $\pdqporacle$ and $\pdqporacle^*$ have the same inputs, we only require their outputs $\{v_i\}_i^\abilityamount$ and $\{v_i^+\}_i^\abilityamount$ are sampled from the same distribution. Consider an arbitrary pair $v_i$ and $v_i^+$ which are the respective measurements of $\ket{\psi_i}$ and $\ket{\psi_i^+}$ where $\ket{\psi_i^+}$ is the resulting state on register $i$ of $C^*$. As the application of $U^*$ is $U_i$ in parallel, only $U_i$ is applied to register $i$ in $C^*$. Therefore we only require that the probability distributions $p$ and $p^*$ of the measurement operators on $\ket{\psi_i}$ and $\ket{\psi_i^+}$ are equivalent.

    Suppose that for $\ket{\psi_i}$, the probability of measuring an arbitrary string $n$ during $M_i$ is $p(n)=a$. For $\ket{\psi_{i}^*}$, we have that the probability of measuring $n^{\otimes \parallelregisters}$, the equivalent of $n$ in $C^*$, is,
    \begin{align*}
        p^*(n) =  d_{i,n} \bra{\psi_{i}^*}\ket{n^{\otimes \parallelregisters}}\bra{n^{\otimes \parallelregisters}}\ket{\psi_{i}^*} = d_{i,n}\;a^{\parallelregisters} = a
    \end{align*}
    As the probability distributions $p$ and $p^*$ are equivalent, the oracles produce equivalent outputs.
\end{proof}

We shall add queries to our model. Let $Q_{i,x}$ the oracle operator applied at step $i$ for input $x$, \queryamount be the number of queries performed and $\parallelregisters = P + Q - i + 1$. In the quantum oracle model with non-collapsing measurements, we insert the oracle gates after the unitary $U_i$ and measurement operator $M_i$. Without loss of generality, we may assume that we don't obtain a non-collapsing measurement at step $i$ if we use $Q_i$ at this step. This enables the circuit to choose when to perform queries and non-collapsing measurements. In $\pdqporacle^*$, this means that we apply a gate $Q_i^*$ in parallel over \parallelregisters registers. You may see examples of the circuits being run by \pdqporacle and $\pdqporacle^*$ in Figure~\ref{figure_pdqp_query}. 

\begin{definition}\label{definition_pdqp_our}
    The \pdqpname query model creates a circuit $C$ based on the input $x$, provides it to \pdqporacle which runs the circuit $C^*$. The query complexity of the model is measured by the number of queries \queryamount (steps of the algorithm which perform a query) and non-collapsing measurements \abilityamount (steps which don't perform a query) in $C$.
\end{definition}

\begin{figure}
    \centering
    \input{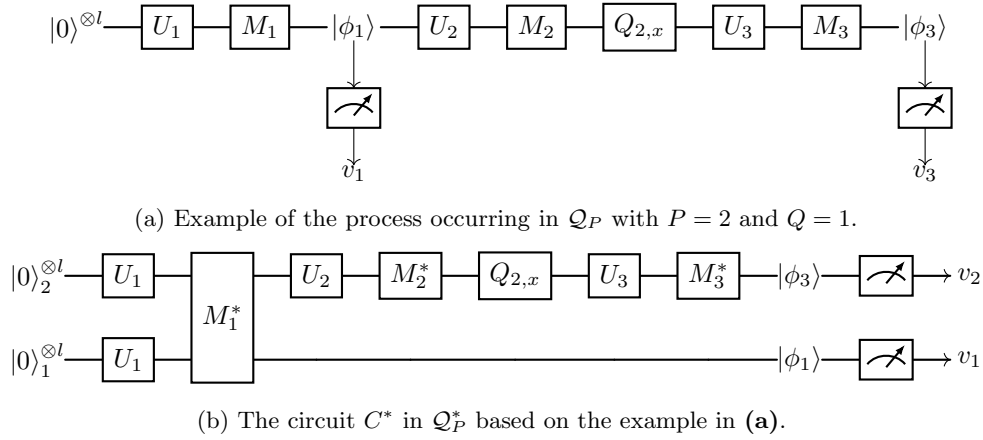}
    \caption{Examples of the quantum oracle model with non-collapsing measurements based on \textbf{(a)} Definition~\ref{definition_pdqp_original} and \textbf{(b)} Definition~\ref{definition_pdqp_our}.}
    \label{figure_pdqp_query}
\end{figure}

By using the definition above, we find that the query complexity of the two models is equivalent.

\begin{lemma}
    The query complexity of the original \pdqpname query model is equivalent to the query complexity of the adjusted \pdqpname query model.
\end{lemma}
\begin{proof}
    Suppose that we query a circuit which performs \queryamount queries and \abilityamount non-collapsing measurements to the original oracle \pdqporacle. By our explanation of the query model above, we know that $\pdqporacle^*$ simulates the queries $Q_i$ by the query gate $Q_i^*$, obtaining the same query complexity.

    On the other hand, suppose we have a circuit that is run by $\pdqporacle^*$. Each unitary $U^*_i$, measurement $M^*_i$ and query $Q^*_i$ may be run by the original oracle \pdqpname by only running the gates on the first register, leading to the same complexity.
\end{proof}

We note that one obtains that the gate complexity of the models is equal in a similar way. Furthermore, it is important to mention that by going from \pdqporacle to $\pdqporacle^*$ and vice-versa under the same input, we adjust the width of the circuits while maintaining equal depth. By the characterization of fidelity using Bhattacharyya coefficients~\cite{fidelity_bhattacharyya_coeffs}, we know that fidelity between states may not decrease after measurement. We find that the same holds for our adjusted measurement gates, allowing us to restrict the power of $M^*$ in this model.

\begin{lemma}\label{lemma_fidelity_partial_measurement}
    Let $(\rho^\prime_{1, i}, \rho^\prime_{2, i})$ be the pair of states with different query oracles at step $i$ of $C^*$ before the application of $M^*_i$. Let $c\in \{1,2\}$ and $\rho_{c, i} = M^*_i \rho^\prime_{c, i} (M^*_i)^\dagger$. Then,
    \begin{align*}
        F(\rho^\prime_{1, i}, \rho^\prime_{2, i}) \leq F(\rho_{1,i}, \rho_{2, i})
    \end{align*}
\end{lemma}
\begin{proof}
    Let $c\in \{1,2\}$. The proof of the statement is done by explicitly defining the states $\rho^\prime_{c, i}$ as the sum of pure states whose purification results in the maximal fidelity, per Uhlmann's theorem. We show that for any such purification, the same purification process increases the fidelity when applied to $\rho_{c, i}$. This lower-bounds $F(\rho_{1,i}, \rho_{2, i})$ by Uhlmann's theorem, proving our statement.

    Let $k$ be an arbitrary integer. We define the state at each register the same due to the multiplicativity of fidelity.
    As the gate $M^*_i$ only affects \parallelregisters registers, we shall ignore the rest. Each of the \parallelregisters registers is the following mixed state,
    \begin{align*}
        \sigma_{c,i} &= \sum_{k\in [k]} b_{c,k} \ket{\lambda_{c,k}}\bra{\lambda_{c,k}}\\
        \ket{\lambda_{c,k}} &= \sum_j \alpha_{c,j, k} \ket{j}_{\mathcal{M}}\ket{\chi_{c,j,k}}
    \end{align*}
    where $b_{c,k}\in \mathbb{R}$, $\alpha_{c,j,k}\in \mathbb{C}$, $\mathcal{M}$ is the register being measured by $M_i^*$ and $\chi_{c,j,k}$ is its associated workspace.
    The definition of $\ket{\lambda_{c,k}}$ in $\sigma_{c,i}$ is such that if we purify it as we do below, we obtain the maximum fidelity as shown by Uhlmann's theorem.
    Let $A = \otimes_j^{\parallelregisters} [k]_j$, $a\in A$ and $b^*_{c,a} = \prod_{r\in a} b_{c,r}$.
    Therefore the entire state $\rho^\prime_{c, i}$ and its purification may be described as,
    \begin{align*}
        \rho^\prime_{c, i} &= \sigma_{c,i}^{\otimes \parallelregisters} = \sum_{a\in A} b^*_{c,a} (\otimes_{s\in a} \ket{\lambda_{c,s}}\bra{\lambda_{c,s}}) \\
        \ket{\psi^\prime_{c,i}} &= \sum_{a\in A} \sqrt{b^*_{c,a}} (\otimes_{s\in a} \ket{\lambda_{c,s}})\ket{s}
    \end{align*}
    Let $a^*_{c,j,a} = \prod_{s\in A} \norm{\alpha_{c, j, s}}_2^2$, $d$ be the probability of measuring $j$ in $\sigma_{c,i}$, $d_{i,j}=\frac{1}{d^{\parallelregisters - 1}}$ and $\alpha^*_{c,j,a} = \prod_{s\in A} \alpha_{c,j,s}$.
    We apply $M^*_i$ to $\rho^\prime_{c, i}$, which adjusts the amplitudes and ensures that the register $\mathcal{M}$ measures to the same value $j$ over all $r_i$ registers. By using the definition of $M^*_i$ from Line~\ref{eq_our_definition_measurement},
    \begin{align*}
        \rho_{c,i} &= \sum_{a\in A, j} a^*_{c,j,a} d_{i,j} b^*_{c,a} (\otimes_{s\in a} \ket{\chi_{c,j,k}}\bra{\chi_{c,j,k}}) \\
        \ket{\psi_{c,i}} &= \sum_{a\in A, j} \alpha^*_{c,j,a} \sqrt{d_{i,j}b^*_{c,a}} (\otimes_{s\in a} \ket{\chi_{c,j,k}})\ket{s}
    \end{align*}
    By using our assumption about $\rho^\prime_{c, i}$, we find,
    \begin{align}
        F(\rho^\prime_{1, i}, \rho^\prime_{2, i}) &= \bra{\psi^\prime_{1,i}}\ket{\psi^\prime_{2,i}} \label{eq_fidelity_1}\\
        &= \sum_{\substack{a_1\in A_1 \\ a_2\in A_2}} \sqrt{b^*_{1,a_1}b^*_{2,a_2}} \prod_{\substack{s_1\in a_1 \\ s_2\in a_2}} \bra{\lambda_{1,s_1}}\ket{\lambda_{2,s_2}} \delta_{s_1,s_2} \\
        &= \sum_{a\in A, j} \sqrt{b^*_{1,a}b^*_{2,a}} \prod_{s\in a} \bra{\chi_{1, j, s}}\ket{\chi_{2,j,s}}\alpha_{1,j,s}\alpha_{2,j,s} \label{eq_fidelity_3}\\
        &\leq \sum_{a\in A, j} \sqrt{b^*_{1,a}b^*_{2,a}} \alpha^*_{1,j,a}\alpha^*_{2,j,a} d_{i,j} \prod_{s\in a} \bra{\chi_{1, j, s}}\ket{\chi_{2,j,s}} \label{eq_fidelity_4}\\
        &= \bra{\psi_{1,i}}\ket{\psi_{2,i}} \label{eq_fidelity_5}\\
        & \leq \bra{\psi^{max}_{1,i}}\ket{\psi^{max}_{2,i}} \label{eq_fidelity_6}\\
        &= F(\rho_{1, i}, \rho_{2, i})
    \end{align}
    Lines~\ref{eq_fidelity_1}-~\ref{eq_fidelity_3} are by our definitions above, Line~\ref{eq_fidelity_4} follows as $d_{i,j}\geq 1$, and Line~\ref{eq_fidelity_6} follows by Uhlmann's theorem where we define $(\psi^{max}_{1,i},\psi^{max}_{2,i})$ as the pair of purifications which maximize $F(\rho_{1,i}, \rho_{2, i})$.
\end{proof}

\section{Lower-bounds on non-collapsing measurements}

In this section we expand the bounds established in~\cite{space_above_bqp}. First, we generalize their lower-bound to a wider selection of problems. Next, we show that the $\Tilde{O}(N^{1/3})$ search algorithm from~\cite{space_above_bqp} is optimal up to logarithmic factors by obtaining a $\Omega(N^{1/3})$ lower-bound.

\subsection{Weighted adversary method}

Our proof closely follows the original proof of the positive-weighted adversary method~\cite{weighted_adversary}. There are two main adjustments. We employ Lemma~\ref{lemma_fidelity_partial_measurement} to argue that using partial measurements won't affect the progress measure. Additionally, we consider the effect of non-collapsing measurements by using the model in Definition~\ref{definition_pdqp_our}. In order to accommodate the second change, we will require the following Proposition.

\begin{proposition}\label{proposition_polynomial_bound}
        For any $k\in \mathbb{N}$ and $r,s\in \mathbb{R}$ where $ 0\leq r \leq 1 $ and $0\leq s \leq 2r$,
        \begin{align*}
            ks \geq r^k - (r-s)^k
        \end{align*}
\end{proposition}
\begin{proof}
    Let $A_{k,s} = ks - r^k + (r-s)^k$. Note that the statement holds when $k=1$. Therefore let $k\geq 2$.

    Assume $s\leq r$. We have,
    \begin{align*}
        A_{n+1,s} - A_{n,s} &= s - r^n(r-1) + (r-s)^n(r-s-1)\\
        &= s -r^ns +r^ns -r^n(r-1) + (r-s)^n(r-s-1)\\ 
        &= s(1-r^n)+((r-s)^n -r^n)(r-1-s)
    \end{align*}
    By the bounds on $r,s$ and the assumption that $s\leq r$, we have that $s(1-r^n)\geq 0$, $(r-1-s)\leq 0$ and $((r-s)^n -r^n)\leq 0$. As $A_{n+1,s} - A_{n,s} \geq 0$, the statement holds. 

    Assume $s>r$. Additionally, without loss of generality, assume $k$ is odd. We find,
    \begin{align*}
        A_{k,s} &= ks-(r^k-(r-s)^k)\\
            &\geq ks - (r^k-(r-2r)^k)\\
            &= ks - 2r^k
    \end{align*}
    As $k\geq 2$ and $s>r$, $A_{k,s}\geq 0$.
\end{proof}

We define the progress measure as,

\begin{align*}
    \Phi(t) = \sum_{(x,y)\in R} w(x,y) F(\rho_{x,t}, \rho_{y,t})
\end{align*}

We bound the progress made at each step using the following lemma.

\begin{lemma}\label{lemma_progress_measure_partial_measurement}
    Let $f$ be a boolean function with a valid weight scheme and its associated maximum load $v_{max}$. Then the progress made by a \pdqpname algorithm each query $t\in [\queryamount]$ is bounded by,
    \begin{align*}
        \Phi(t-1)-\Phi(t)\leq \abilityamount\, v_{max} \Phi (0)
    \end{align*}
\end{lemma}

\begin{proof}
    By the unitary invariance of fidelity, we will ignore the application of unitaries. Similarly, by Lemma~\ref{lemma_fidelity_partial_measurement}, we won't consider partial measurements as their effect may not decrease the progress measure. Therefore, without loss of generality, we assume our states are always pure. Additionally, we assume that during query $t$, we have $D$ non-collapsing measurements left and have already performed $d=P-D$.

    Let $\ket{\psi_{x,t}}$ denote the state of the algorithm on input $x$ at query $t$ and $\ket{\psi_{x,t}^\prime}$ denote the same before the application of the oracle. Additionally, let $\ket{\psi_{x,t}^*}$ be the state during query $i$ of $C^*$ as defined in Definition~\ref{definition_pdqp_our}, while $\ket{\psi_{x,t}^{\prime *}}$ is the state in $C^*$ before the query. Using these definitions, we obtain the following identity,

    \begin{align}
        \Phi(t-1) - \Phi(t) &= \sum_{(x,y)\in R} w(x,y) F( \psi_{x, t-1}^*,\psi_{y, t-1}^* ) - \sum_{(x,y)\in R} w(x,y) F(\psi_{x, t}^*,\psi_{y, t}^* )\\
        & \leq \sum_{(x,y)\in R} w(x,y) \abs{ \bra{\psi_{x,t-1}^*}\ket{\psi_{y,t-1}^*}  -  \bra{\psi_{x,t}^*}\ket{\psi_{y,t}^*}}\\ 
        &= \sum_{(x,y)\in R} w(x,y) \abs{\bra{\psi_{x,t}^{\prime *}}\ket{\psi_{y,t}^{\prime *}}  - \bra{\psi_{x,t}^*}\ket{\psi_{y,t}^*}} \label{eq_progress_simplification}
    \end{align}
    
    Let $\ket{\psi_{x,t}} = \sum_{i} \alpha_{x,i} \ket{i}_\mathcal{Q}\ket{\phi_{x,i}}_\mathcal{W}$ where $\mathcal{Q}$ is the query register and $\mathcal{W}$ is the workspace register associated with that query register. As we have $D$ non-collapsing measurements left to perform, an oracle call will affect the $D$ copies of the state labeled by $\mathcal{O}$ and neglect the remaining $d$.
    \begin{align*}
        \ket{\psi_{x,t}^*} = \ket{\psi_{x,1}}\otimes..\otimes\ket{\psi_{x,d}}\otimes (\ket{\psi_{x,t}}_\mathcal{O})^{\otimes D}
    \end{align*}

    For the sake of simplicity, let,
    \begin{align*}
        \hat{S}_i &= \sum_{i} \alpha_{x,i}\alpha_{y,i}\bra{\phi_{x,i}}\ket{\phi_{y,i}}\\
        \hat{S}_{x(i)=y(i)} &= \sum_{i:x(i)=y(i) } \alpha_{x,i}\alpha_{y,i}\bra{\phi_{x,i}}\ket{\phi_{y,i}}\\
        \hat{k} &= \sum_{i: x(i)\neq y(i)}\alpha_{x,i}\alpha_{y,i}\bra{\phi_{x,i}}\ket{\phi_{y,i}}
    \end{align*}
    We find that,
    \begin{align}
        \Phi(t-1) - \Phi(t) \leq 
        &\sum_{(x,y)\in R} w(x,y) \left| \bra{\psi_{x,t}^{\prime *}}\ket{\psi_{y,t}^{\prime *}} - \bra{\psi_{x,t}^*}\ket{\psi_{y,t}^*} \right|\label{eq_weight_1}\\
        = &\sum_{(x,y)\in R} w(x,y)\abs{\bra{\phi_{x,1}}\ket{\phi_{y,1}}...\bra{\phi_{x,d}}\ket{\phi_{y,d}}} \notag\\
        &\qquad\abs{ \hat{S}_{i_1}..\hat{S}_{i_D} - \left(\hat{S}_{x(i_1)=y(i_1)} - \hat{k}_1\right)...\left(\hat{S}_{x(i_D) = y(i_D)} - \hat{k}_D\right)} \label{eq_weight_2}\\
        \leq &\sum_{(x,y)\in R} w(x,y)\abs{\hat{S}_i ^D - \left(\hat{S}_i - 2\hat{k}\right)^D } \label{eq_weight_3}\\
        \leq & \sum_{(x,y)\in R} 2 w(x,y) D \hat{k} \label{eq_weight_4}\\
        = & D \sum_{\substack{(x,y)\in R \\ i: x(i)\neq y(i)}} 2 w(x,y) \left|\alpha_{x,i} \right| \left|\alpha_{y,i}\right|\label{eq_weight_5}\\
        \leq & P\; v_{max}\Phi(0)\label{eq_weight_6}
    \end{align}
    Line~\ref{eq_weight_1} was shown on Line~\ref{eq_progress_simplification}, Line~\ref{eq_weight_2} uses the multiplicativity of fidelity, Line~\ref{eq_weight_3} holds as the inner product of two states is bounded by 1, Line~\ref{eq_weight_4} uses Proposition~\ref{proposition_polynomial_bound} by setting $r=\hat{S}_i$ and $s=2\hat{k}$, and Line~\ref{eq_weight_6} is by Proposition~\ref{proposition_weight_identity} and $D\leq \abilityamount$.
    \end{proof}

    Lemma~\ref{lemma_progress_measure_partial_measurement} is sufficient to obtain the bound below.

    \begin{lemma}\label{lemma_weighted_adversary_partial_measurement}
        Let $f$ be a boolean function with a weight scheme and its associated maximum load $v_{max}$. In a \pdqpname setting, the following holds,
        \begin{align*}
            QP = \Omega\left(\frac{1}{v_{max}}\right)
        \end{align*}
    \end{lemma}
    \begin{proof}
        Assuming success probability of $1-\epsilon$, we have that $\Phi(T) \leq 2\sqrt{\epsilon (1-\epsilon)} \Phi(0)$. By Lemma~\ref{lemma_progress_measure_partial_measurement},
        \begin{align*}
            Q\geq \frac{\Phi(0)-\Phi(T)}{\Phi(t-1)-\Phi(t)} \geq \frac{1-2\sqrt{\epsilon(1-\epsilon)}}{P v_{max}}
        \end{align*}
\end{proof}

\subsection{Lower-bound on search}

We use the proof of a $\Omega(N^{1/3})$ lower-bound on \pdqpname without collapsing measurements from~\cite{space_above_bqp} and apply Lemma~\ref{lemma_fidelity_partial_measurement} in order to introduce partial measurement.

\begin{lemma}\label{lemma_lowerbound_search}
    Any algorithm solving the search problem in \pdqpname setting using \queryamount queries and and \abilityamount non-collapsing measurements satisfies,
    \begin{align*}
        \queryamount+\abilityamount = \Omega(N^{1/3})
    \end{align*}
\end{lemma}
\begin{proof}

Consider a \pdqpname algorithm where $Q=P = o(N^{1/3})$. We compare the cases when there are no marked items versus when $x$ is a marked item. Let their respective quantum states at step $t$ be $\rho_t, \rho_t^x$ and $\psi_t$ and $\psi_t^x$ be states of the same algorithm without partial measurement. By Lemma~\ref{lemma_hybrid_argument}, we have,

\begin{align*}
    \sum_x \norm{\psi_t-\psi_t^x}^2_2 \leq 4Q^2
\end{align*}

By summing over $P$ non-collapsing measurements and using an averaging argument, there exists an $x$ such that,

\begin{align*}
\sum_t \norm{\psi_t - \psi_t^x}^2_2 \leq \frac{4PQ^2}{N}
\end{align*}

By our assumption on the number of queries, this norm must be constant. Let's say that $\norm{\psi_t-\psi_t^x}_2^2 \leq 0.01$. Let $\Phi=\otimes_t \rho_t$ represent all the states before measurement and $V$ the distribution of their measurements. In order to notice the presence of a marked item, we require $|V-V_x|_1\geq \Omega(1)$. Therefore,

\begin{align}
    \Omega(1)\leq |V-V_x|_1 &\leq \norm{\Phi - \Phi_x} \label{eq_search_1}\\
    &\leq \sqrt{1-F(\Phi, \Phi_x)^2}\label{eq_search_2}\\
    & =\sqrt{1-|\Pi_t F(\rho_t, \rho_t^x)|^2}\label{eq_search_3}\\
    &\leq \sqrt{1-|\Pi_t \bra{\psi_t}\ket{ \psi_t^x}|^2}\label{eq_search_4}\\
    &\leq \sqrt{1-\Pi_t e^{-\norm{\psi_t-\psi_t^x}^2_2}}\label{eq_search_5}\\
    &\leq \sqrt{1 - e^{-\frac{4PQ^2}{N}}}\label{eq_search_6}\\
    &=o(1)
\end{align}

Line~\ref{eq_search_2} follows by the definition of trace norm, Line~\ref{eq_search_4} holds as by Lemma~\ref{lemma_fidelity_partial_measurement}, we may assume we did not perform partial measurements, and Line~\ref{eq_search_5} uses the inequality 
\begin{align*}
    |\bra{\psi_t}\ket{\psi_t^x}|\geq \text{Re}(\bra{\psi_t}\ket{\psi_t^x})\geq e^{-\norm{\psi_t-\psi_t^x}^2_2}
\end{align*}
which holds as $1-x\geq e^{-2x}$ when $0\leq x\leq 0.01$.
\end{proof}

\section{Non-collapsing measurements and non-adaptive queries}

\subsection{Non-adaptive query model}\label{subsection_nonadaptive_model}

In this section, we describe the non-adaptive query model and prove its lower-bounds. Informally, the model performs all of its queries at the start of the computation, meaning that they are independent of the input $x$. We obtain a trade-off between \queryamount and \abilityamount.

The notation and proof closely follows that of~\cite{nonadaptive_weighted_adversary}. Assume we perform \queryamount non-adaptive queries. Given an oracle operator $O_x$ for an input $x$, we define the query as,
\begin{align*}
    O_x^{\queryamount} \ket{i_1, b_1, ... i_Q, b_Q} \rightarrow (-1)^{s} \ket{i_1, b_1, ... i_Q, b_Q}
\end{align*}

where $s = \sum_{j\in [Q]} x(i_j)\cdot b_j$. The algorithm may be described as follows,
\begin{align*}
    \measurementgate{\abilityamount}\unitarygate{\abilityamount}... \measurementgate{1}\unitarygate{1}O_x^{\queryamount}\unitarygate{0}\ket{0}
\end{align*}

Assume we perform a non-collapsing measurement after every unitary except $\unitarygate{0}$. The label \nonadaptivelabel describes a variable in the non-adaptive query setting. For example, $\abilityamount^{\nonadaptivelabel}$ describes the number of non-collapsing measurements performed. Next, we define a variation of the weight scheme from Definition~\ref{definition_weight_scheme} for non-adaptive queries. The main idea of the proof is lower-bounding the minimum size of the initial parallel query by arguing that a smaller query would require at least two rounds of parallel queries. We use the superscript $k$ to denote a variable for the parallel query of size $k\in \mathbb{N}$.

\begin{definition}[Weight scheme of~\cite{nonadaptive_weighted_adversary}]\label{definition_nonadaptive}
    Consider an arbitrary boolean function $f$ and an input $x$. We define ${}^k x = x^k$, ${}^k f({}^k x) = f(x)$ and an index-tuple $I = (i_1, ... i_k)$. Therefore $x(I) = (x(i_1),...,x(i_k))$. 
    
    Let $w, wt$ and $v$ be the weight variables as in Definition~\ref{definition_weight_scheme} for $f$. Then $W, WT$ and $V$ are their variations for ${}^k f$. We let $W({}^k x, {}^k y) = w(x,y)$.
\end{definition}

We will use the following lemma~from~\cite{nonadaptive_weighted_adversary},

\begin{lemma}[Lemma~2 in~\cite{nonadaptive_weighted_adversary}]\label{lemma_nonadaptive_average_weights}
    Let $x$ be an input and $w$ a weight function. For any $k\in \mathbb{N}$ and an index-tuple $I=(i_1,..,i_k)$,
    \begin{align*}
        \frac{WT({}^kx)}{V({}^kx, I)} \geq \frac{1}{k} \min_{j\in [k]}\frac{wt(x)}{v(x,i_j)}
    \end{align*}
\end{lemma}

\subsection{Weighted adversary method with non-adaptive queries}

We prove the results on non-adaptive queries using the technique from~\cite{nonadaptive_weighted_adversary}, resulting in the lemma~below.

\begin{lemma}\label{lemma_nonadaptive}
    Consider a function $f$ with a weight scheme $w$. Then any \pdqpnaqname algorithm with \queryamount queries and \abilityamount non-adaptive measurements satisfies,
     \begin{align*}
        \queryamount\abilityamount = \Omega \left( \max_w \max_{s\in \{0,1\}} \min_{\substack{x,i \\ f(x)=s}} \frac{wt(x)}{v(x,i)} \right)
    \end{align*}
\end{lemma}

In order to establish Lemma~\ref{lemma_nonadaptive}, we will first establish the following result which connects Lemma~\ref{lemma_weighted_adversary_partial_measurement} and weights for non-adaptive inputs.

\begin{lemma}\label{lemma_nonadaptive_weights}
    Let $w$ be a weight function for $f$. Then $W$ is a valid weight function for ${}^kf$ and any algorithm solving ${}^kf$ using \queryamount queries and \abilityamount non-collapsing measurements satisfies,
    \begin{align*}
        \queryamount\abilityamount \geq C_{\epsilon} \min_{{}^kx, {}^yx, I} \left(\frac{WT({}^kx) WT({}^ky)}{V({}^kx, I)V({}^ky,I)}\right)^{1/2}
    \end{align*}
    where $C_{\epsilon} = \sqrt{\frac{1-2\sqrt{\epsilon (1-\epsilon)}}{2}}$.
\end{lemma}
\begin{proof}
    By Definition~\ref{definition_nonadaptive}, as $W({}^kx, {}^ky)=w(x,y)$, we may apply Lemma~\ref{lemma_weighted_adversary_partial_measurement} to obtain the bound.
\end{proof}

We will prove the following lemma:
\begin{lemma}\label{lemma_nonadaptive_final}
    Let $w$ be a valid weight function and $x,y$ such that $f(x)\neq f(y)$. Then,
    \begin{align*}
    \queryamount^{\nonadaptivelabel} \abilityamount^{\nonadaptivelabel} \geq C^2_{\epsilon} \min_{x,y}\; \max\; \left(\min_i \frac{wt(x)}{v(y,i)}, \min_i \frac{wt(y)}{v(x,i)}\right)
    \end{align*}
\end{lemma}
\begin{proof}
    Let $k\in \mathbb{N}$ be such that
    \begin{align}
        k < C^2_{\epsilon} \min_{x,y} \max \left(\min_i \frac{wt(x)}{v(x,i)}, \min_i \frac{wt(y)}{v(y,i)}\right)
    \end{align}
    We shall show that performing $k$ non-adaptive queries would require us to repeat the computation more than once, implying that $\queryamount^{\nonadaptivelabel}\abilityamount^{\nonadaptivelabel}>k$. By using Lemma~\ref{lemma_nonadaptive_average_weights} and the fact that $WT({}^kx)\geq V({}^kx,I)$ for any $x$ and $I$, we have that,
    \begin{align*}
        \frac{WT({}^kx)}{V({}^kx, I)}\frac{WT({}^ky)}{V({}^ky, I)} &\geq \max \left(\frac{WT({}^kx)}{V({}^kx, I)},\frac{WT({}^ky)}{V({}^ky, I)}\right)\\
        &\geq \frac{1}{k} \max \left(\min_{j\in [k]}\frac{wt(x)}{v(x,i_j)}, \min_{l\in [k]}\frac{wt(y)}{v(y,i_l)}\right)
    \end{align*}
    Therefore we have that,
    \begin{align*}
        \min_{{}^kx, {}^ky, I} \frac{WT({}^kx)}{V({}^kx, I)}\frac{WT({}^ky)}{V({}^ky, I)} &\geq \frac{1}{k} \min_{{}^kx, {}^ky, I} \max \left(\min_{j\in [k]}\frac{wt(x)}{v(x,i_j)}, \min_{l\in [k]}\frac{wt(y)}{v(y,i_l)}\right)\\
        &\geq \frac{1}{k} \min_{x, y} \max \left(\min_i \frac{wt(x)}{v(x,i)}, \min_i \frac{wt(y)}{v(y,i)}\right)\\
        & \geq \frac{1}{C^2_{\epsilon}}
    \end{align*}
    Where we applied our assumption on $k$. Therefore according to Lemma~\ref{lemma_nonadaptive_weights}, $\queryamount^{\nonadaptivelabel}\abilityamount^{\nonadaptivelabel} > 1$, proving the lemma.
\end{proof}

Lemma~\ref{lemma_nonadaptive_final} directly implies Lemma~\ref{lemma_nonadaptive}.

\section{Lower-bounding quantum computation with copies}

In this section, we apply the ideas from Lemma~\ref{lemma_weighted_adversary_partial_measurement} to \copybqpname, quantum computation which is able to create a copy of any state. Another way to consider \copybqpname is as the computational class representing a world where the no-cloning theorem does not hold. The point of this section is to explicitly showcase the power difference between non-collapsing measurements and copying a state. As most steps in our analysis are identical, we only focus on the differences.

\subsection{Computational model with copies}

The key difference between \pdqpname and \copybqpname is that the latter has the ability to apply unitaries to the copies after creating them, while \pdqpname must immediately collapse them. As these copies may become entangled, purifying the algorithm in a similar manner as in Definition~\ref{definition_pdqp_our}, meaning explicitly creating copies of each state in order to analyze it, will require additional registers. Therefore, as we explain below, instead of using $O(\abilityamount)$ registers as in \pdqpname, we will need $O(2^{\abilityamount})$ registers.\footnote{\abilityamount denotes either the number of non-collapsing measurements or copies. This will be obvious from context.}

\begin{figure}
    \centering
    \input{tikz_cbqp}
    \caption{Example of purifying a \copybqpname circuit.}\label{figure_cbqp}
\end{figure}

Assume we perform \abilityamount copies in total. Let $r_0$ be the starting register and for $i>0$, $r_i$ be its copies at step $i$. Without loss of generality, assume we only copy and apply the oracle to $r_0$. The gate $C_i$ symbolizes a copy gate, which creates a copy, allowing us to apply unitaries across all registers $r_j$ for $j\leq i$. Prior to applying the $C_i$ gate, we must ensure that $r_i$ and $r_0$ are copies, meaning that similarly to Definition~\ref{definition_pdqp_our}, we apply unitary $U$ in parallel to both. Finally, we account for the effect of partial measurement by using a partial measurement gate $M^*_j$, ensuring that partial measurement results before applying $C_i$ are the same across all registers. When adding queries to this model, we follow the same steps as in Definition~\ref{definition_pdqp_our}.

Finally, as each copy must mimic the register it is copying, we need to add further purification registers $p_{i,j}$ for each copy $i$. These registers mimic the process of the register until it has been copied. As each copy may be entangled, with each additional copy, the amount of purification registers doubles. In order to describe copy $i$ of the computation, we require $2^i$ registers including the register $r_i$. Therefore the total number of registers is $r^{\abilityamount}$. An example of a circuit with the ability to copy and its purification is in Figure~\ref{figure_cbqp}.

\subsection{Adversary lower-bound}

We prove a version of Lemma~\ref{lemma_progress_measure_partial_measurement} for $\mathsf{CBQP}$.

\begin{lemma}\label{lemma_copy}
    Let $f$ be a boolean function with a valid weight scheme and its associated maximum load $v_{max}$ as in Definition~\ref{definition_weight_load}. Then the progress made by a \copybqpname algorithm each query $t$ is bounded by,
    \begin{align*}
        \Phi(t-1)-\Phi(t)\leq 2^{\abilityamount} v_{max} \Phi (0)
    \end{align*}
    where \abilityamount is the number of copies.
\end{lemma}
\begin{proof}
    First, we shall ignore the effect of unitaries and partial measurement by the same arguments and definitions of variables as in Lemma~\ref{lemma_progress_measure_partial_measurement}. In this case $D\leq 2^{\abilityamount}$. We obtain the following result,

    \begin{align}
        \Phi(t-1) - \Phi(t) \leq 
        &\sum_{(x,y)\in R} w(x,y) \abs{ \bra{\psi_{x,t}^{\prime *}}\ket{\psi_{y,t}^{\prime *}} - \bra{\psi_{x,t}^*}\ket{\psi_{y,t}^*}}\label{copy_weight_1}\\
        \leq &\sum_{(x,y)\in R} w(x,y)\abs{\hat{S}_i^D - \left(\hat{S}_i - 2\hat{k}\right)^D} \label{copy_weight_2}\\
        \leq & 2^{\abilityamount}v_{max}\Phi(0)\label{copy_weight_3}
    \end{align}
\end{proof}

We find the following relationship between the number of queries and copies.

\begin{lemma}\label{lemma_copy_adversary}
    Let $f$ be a boolean function with a weight scheme and its associated maximum load $v_{max}$. A $\mathsf{CBQP}$ algorithm with \queryamount queries and \abilityamount copies satisfies,
        \begin{align*}
            \queryamount 2^{\abilityamount} = \Omega\left(\frac{1}{v_{max}}\right)
        \end{align*}
\end{lemma}

Lemma~\ref{lemma_copy_adversary} implies that the algorithm in~\cite{grover_search_no_signaling_principle} to solve search is optimal and explains the discrepancy between \copybqpname and \pdqpname.

\section{Bounds on problems with non-collapsing measurements}

\subsection{Proofs of applications of the lower-bounding techniques}\label{subsection_proof_application}

We shall apply Lemmata~\ref{lemma_weighted_adversary_partial_measurement} and~\ref{lemma_nonadaptive_final} on concrete problems in order to compare the two models. First, let us note that the positive-weighted adversary technique may be reduced to the basic adversary technique by setting the weights $w(x,y)$ to 1.

\begin{corollary}\label{corollary_basic_adversary}
    Let $f:\{0,1\}^n\rightarrow \{0,1\}$ be a boolean function, $X,Y$ be two subsets of \kinputs{1} and \kinputs{0} respectively and $R$ a relation $R\subseteq X \times Y$. Furthermore, let $m,m^\prime, l, l^\prime$ be the same as in Result~\ref{result_adversary}. Then a \pdqpname query algorithm solving $f$ with \queryamount queries and \abilityamount non-collapsing measurements satisfies,
    \begin{align*}
        \queryamount\abilityamount = \Omega\left( \sqrt{\frac{mm^\prime}{ll^\prime}} \right)
    \end{align*}
    Furthermore, a \pdqpname query algorithm with \emph{non-adaptive} queries solving $f$ with $\queryamount^{\nonadaptivelabel}$ queries and $\abilityamount^{\nonadaptivelabel}$ non-collapsing measurements satisfies,
    \begin{align*}
        \queryamount^{\nonadaptivelabel} \abilityamount^{\nonadaptivelabel} = \Omega \left( \max\left\{\frac{m}{l}, \frac{m^\prime}{l^\prime}\right\}\right)
    \end{align*}
\end{corollary}
\begin{proof}
    In both Lemmata~\ref{lemma_weighted_adversary_partial_measurement} and~\ref{lemma_nonadaptive_final}, let $w(x,y) = 1$. Then $w(x)\geq m$ or $m^\prime$ and $v(x,i)\leq l$ or $l^\prime$ depending on whether $x\in X$ or $x\in Y$.
\end{proof}

Let us obtain the values $m,m^\prime, l, l^\prime$ for the search, majority, parity, collision and element distinctness problems. For each problem, we shall define the sets $X, Y$ and their relation $R$ which maximizes $m, m^\prime$ while minimizing $l, l^\prime$. In addition to that, we also describe their algorithm. As has been done during the entirety of the document, we let $N$ be the size of the input. Finally, we note that a summary of the results described below may be found in Table~\ref{table_bounds}.

\textbf{Search problem}. We know that the query complexity of the search problem in \pdqpname is $\Tilde{\Theta}(N^{1/3})$ due to the algorithm in~\cite{space_above_bqp}, while our Lemma~\ref{lemma_lowerbound_search} proves that this algorithm is tight up to logarithmic factors. However, the same technique does not apply to \pdqpname with non-adaptive queries. We may split the inputs for search as follows:
\begin{itemize}
    \item Let $X=\{0\}$ as $f(x)=0$ if and only if $x=0$.
    \item Let $Y=\{y: |y|=1\}$ as these are the only values which only differ by 1 character from any input in $X$.
\end{itemize}

Therefore we find that $m=N$ as we may choose any index $i\in[N]$ to go from $X$ to $Y$. Furthermore, $m^\prime = l = l^\prime = 1$ as each of these connections between $X$ and $Y$ only occur on one index. By Corollary~\ref{corollary_basic_adversary}, the query complexity of the search problem in the \pdqpname setting with non-adaptive measurements is $\Omega(\sqrt{N})$. This is tight as one may run an algorithm where we partition $[N]$ into $\sqrt{N}$ sets and query each separately. Afterwards, we may perform non-collapsing measurements on all partitions at once which, with $\Tilde{O}(\sqrt{N})$ would return a marked item with high probability if there exists one.

\textbf{Majority problem}. To lower-bound the query complexity of the majority problem, we shall split the sets $X$ and $Y$ ones whose Hamming weight differs by one. Due to the structure of the problem, this is only possible when the Hamming weight is around $N/2$. Therefore let $X=\{x:|x|=N/2\}$, $Y: \{y: |y|=N/2+1\}$ and $(x,y)\in R$ if and only if $|\{x(i): x(i)\neq y(i)\}| = 1$. This means that $m=N/2$, $m^\prime=N/2+1$ and $l=l^\prime = 1$. Therefore, by Corollary~\ref{corollary_basic_adversary}, the lower-bounds for \pdqpname is $\Omega(\sqrt{N})$ in both \pdqpname cases, regardless of the adaptivity of queries. 

Similarly to the non-adaptive algorithm for search, the best algorithm for the problem is performing $\sqrt{N}$ queries partitions of the input, each of size $\sqrt{N}$. Afterwards, we may perform $\Tilde{O}(\sqrt{N})$ non-collapsing measurements. By the coupon collector's problem, we will obtain all values with high probability, meaning the complexity of the problem is $\Tilde{O}(\sqrt{N})$. Finally, as both the same algorithm for the upper-bound and pair of sets $(X,Y)$ for the lower-bound may be applied for the parity problem, we obtain the same result.

\textbf{Element Distinctness problem}. To analyze the element distinctness problem, we let $X$ be the set of injective functions and $Y$ the set of functions which are almost injective. By almost injective, we mean that there exists exactly one pair of inputs $(x,x^\prime)$ such that $f(x)=f(x^\prime)$. Then we may define the relation between these sets of inputs as those which only require one change of outputs. Explicitly, $(f,f^\prime)\in R$ if and only if $|\{x: f(x)\neq f^\prime(x)\}| = 1$. Thus every injective function may become almost injective by changing any one of its image to another image which already has a preimage. Therefore $m=N$ and $l=1$. On the other hand, one may only choose one of two indices if they want to turn an almost injective function into an injective function by changing a single index. Therefore $m^\prime = 2$ and $l^\prime = 1$. Thus we find that the element distinctness problem in the \pdqpname query models with and without adaptive queries have a $\Omega(N^{1/4})$ and $\Omega(\sqrt{N})$ lower-bound respectively. Finally, the algorithm is the same as for non-adaptive search and the majority problem, giving us a query bound of $\Tilde{O}(\sqrt{N})$.

\subsection{Polynomial method lower-bounding technique}\label{subsection_polynomial_method}

We applied a \bqpname lower-bounding technique to \pdqpname, but it is not tight in all situations. This may be due to the fact that the lower-bound establishes a $\queryamount\abilityamount$ relationship instead of $\queryamount+\abilityamount$. We note that the lower-bound on search in~\cite{space_above_bqp} had the same issue. It would be interesting if one was able to apply the polynomial method to the setting. One approach towards this would be going in the spirit of Definition~\ref{definition_pdqp_our}. Let $p_i(x)$ be the polynomial associated with the amplitudes of the state at register $i$. One may attempt to restrict each $p_i(x)$ in two ways. First, for each $i$, one would ensure that $p_i(x)$ builds on $p_{i-1}(x)$ in that if we omitted the last query from $p_i(x)$, we would obtain $p_{i-1}(x)$. Second, the polynomial $p(x)$ describing the amplitude of the entire state is a product of all $p_i(x)$ with a caveat. Without partial measurement, $p(x)=\prod_i p_i(x)$. However, after partial measurement, $p_i(x)$ would become the sum of several polynomials, $p_i(x)=\sum_j p_i^j(x)$ where the superscript $j$ signals a measurement result. In order to ensure partial measurement remains consistent along all polynomials, we would require $p(x)=\sum_j \prod_i p_i^j(x)$. We note that a list of other open problems we found interesting may be found in Subsection~\ref{subsection_open_problems}.

\bibliography{main}
\end{document}